\documentclass[a4paper,fleqn]{cas-sc}

\usepackage{amsmath,amssymb,amsthm,mathtools}
\usepackage{graphicx}
\usepackage{float}
\usepackage{capt-of}
\usepackage{enumitem}
\usepackage{xcolor}
\usepackage{url}
\usepackage{booktabs}
\usepackage{mathrsfs}
\usepackage[title]{appendix}
\usepackage{algorithm}
\usepackage{algorithmicx}
\usepackage{algpseudocode}
\usepackage{listings}
\usepackage[numbers]{natbib}
\newcommand{\headinggap}{\vspace{0.55\baselineskip}}

\newtheorem{theorem}{Theorem}[section]
\newtheorem{lemma}[theorem]{Lemma}

\newtheorem{corollary}[theorem]{Corollary}
\newtheorem{observation}[theorem]{Observation}
\newtheorem{construction}[theorem]{Construction}
\theoremstyle{definition}
\newtheorem{definition}[theorem]{Definition}
\theoremstyle{remark}
\newtheorem{remark}[theorem]{Remark}

\DeclareMathOperator{\mcr}{mcr}
\DeclareMathOperator{\ecc}{ecc}

\begin{document}
\let\WriteBookmarks\relax
\def\floatpagepagefraction{1}
\def\textpagefraction{.001}

\shorttitle{Efficient \texorpdfstring{$k$}{k}-limited broadcast domination}
\shortauthors{Bharadwaj and A. Senthil Thilak}

\title[mode=title]{A note on efficient \texorpdfstring{$k$}{k}-limited broadcast domination in graphs}

\author[1]{Bharadwaj}
\ead{bwajhsvj@gmail.com}

\author[2]{{\color{black}A. Senthil Thilak}}[orcid=0000-0001-5793-5614]
\ead{thilak@nitk.edu.in, asthilak@gmail.com}

\affiliation[1]{organization={Department of Mathematical and Computational Sciences, National Institute of Technology Karnataka},
                 city={Surathkal},
                 postcode={575025},
                 country={India}}

\begin{abstract}
An efficient $k$-limited dominating broadcast, or  $k$-ELDB, is a  $k$-limited broadcast in which every vertex is dominated exactly once. This notion brings together efficient domination and limited broadcast domination in a common framework. For a graph  $G$, we write  $\mcr(G)$ for the smallest integer  $k$ for which  $G$ admits a  $k$-ELDB. For an admissible value  $k\ge \mcr(G)$, we denote by  $\gamma_{ebk}(G)$ the minimum cost of a  $k$-ELDB on  $G$, and is called the $k$-efficient broadcast domination number of $G$.

In this paper, we study these parameters from an algorithmic perspective with complexity analysis. We develop a dynamic programming algorithm for trees which, for fixed  $k$, computes  $\gamma_{ebk}(T)$ and thereby obtains a polynomial-time procedure for determining  $\mcr(T) $. In contrast, we prove that, for every fixed integer  $k\ge 1$, deciding whether a graph admits a  $k$-ELDB is NP-complete for arbitrary graphs. These results place efficient limited broadcast domination in a natural complexity framework, with trees forming a tractable class and arbitrary graphs remaining computationally hard.
\end{abstract}


\maketitle
\par

\section{Introduction}
\headinggap
Domination is one of the central notions in graph theory. A set
$S\subseteq V(G)$ is a dominating set of a graph $G$ if every vertex of
$G$ belongs to the closed neighborhood of some vertex of $S$. A dominating
set is called \emph{efficient} if each vertex of $G$ is dominated by exactly one
vertex of $S$. Equivalently, the closed neighborhoods of the vertices of
$S$ form a partition of $V(G)$. It has been
studied under several names, including perfect codes and efficient domination
\cite{bange1988,haynes1998}.

Broadcast domination offers a different way to extend the classical domination problem, in which a vertex may dominate vertices at distances greater than one, according to the strength (or cost) assigned to each vertex in the dominating set. The notion was introduced by Erwin~\cite{erwin2001cost} and was further developed by Dunbar et al.~\cite{dunbar2006broadcasts}. The usual optimization problem is to minimize the total broadcast cost, namely the sum of the strengths assigned to all vertices. Heggernes and Lokshtanov showed that ordinary broadcast domination can be solved in polynomial time for arbitrary graphs~\cite{heggernes2006}. Broadcast parameters have been studied extensively on trees and related families; see, for example, \cite{cockayne2011broadcasts,herke2009dominating,herke2009radial,mynhardt2013class,lunney2015more}. The limited version, in which all broadcast strengths are bounded above by a fixed integer $k$, has also been studied; see, for instance, \cite{caceres2018}.

The notion studied here lies between these two directions. Let $G$ be a
connected graph and let $k\ge 1$. A $k$-limited broadcast is a function
$f:V(G)\to \{0,1,\dots,k\}$. A vertex  $v $ with  $f(v)>0 $ is called a \emph{broadcasting vertex}, or simply a \emph{broadcaster}. In broadcast domination, each vertex $v$ with $f(v)\geq 0 $ dominate every vertex  $u\in V(G) $ satisfying
 $
d(u,v)\le f(v)
 $. 
 The broadcast $f$ is dominating if every vertex of $G$ is dominated by at least one broadcaster. Its cost is $\sum_{v\in V(G)}f(v)$. A $k$-limited dominating broadcast is called an efficient $k$-limited dominating broadcast, or a $k$-ELDB, if every vertex of $G$ is dominated by exactly one broadcasting vertex.

We denote by $\mcr(G)$ the least integer $k\ge 1$ for which $G$ admits a
$k$-ELDB. Thus $\mcr(G)=1$ precisely when $G$ is efficiently dominatable
in the classical sense. For a fixed admissible value of $k$, that is, for
$k\ge \mcr(G)$, the $k$-efficient broadcast domination number of $G$, denoted by $\gamma_{ebk}(G)$ is the minimum cost of a
$k$-ELDB on $G$. 

The open problem posed in the broadcast domination chapter of
\cite{haynes2021structures}, namely, "What is the smallest value of $k$ for which a graph $G$ has an
efficient $k$-limited dominating broadcast?"
 serves as the main motivation for this article. The parameter $\mcr(G)$ captures
exactly this question and places efficient domination and limited broadcast
domination into a common framework.

The main contributions of this article are the following:
\begin{itemize}[leftmargin=2em]
    \item We prove that, for every fixed integer $k\ge 1$, deciding whether a
    graph admits a $k$-ELDB is NP-complete for arbitrary graphs. The reduction is
    from Exact $1$-in-$3$ SAT. Consequently, computing $\mcr(G)$ is
    NP-hard for arbitrary graphs.

    \item For every integer $k\ge 1$, constructively we prove the existence of a graph $G$ such that
    $\mcr(G)=k$. This shows that each $k\in Z^+$ occur as the minimum
    covering radius of some graph.
    \item We present a polynomial-time dynamic programming algorithm for trees.
    For a fixed integer $k \geq 1$, the algorithm computes $\gamma_{ebk}(T)$,
    constructs an optimal $k$-ELDB, and decides whether a tree $T$ admits
    a $k$-ELDB. As a consequence, $\mcr(T)$ can be computed in polynomial
    time.
\end{itemize}

\section{A dynamic programming algorithm for trees}
\headinggap
For a fixed integer $k\ge 1$, the $k$-ELDB decision problem asks whether a given graph $G$ admits an efficient $k$-limited dominating broadcast. In this section, we study this problem on trees and show that it is solvable in polynomial time. The key idea is that, once a tree is rooted, each subtree navigates the rest of the tree only through its root. This makes it possible to summarize the effect of broadcasts inside the subtree and then combine the child subtrees in a bottom-up approach.

Throughout this section, assume $T$ to be a tree rooted at a vertex $r$. For a
vertex $v$, let $C(v)$ denote the set of children of $v$, let $T_v$ be
the subtree of $T$ rooted at $v$, and let $\ecc_T(v)$ be the eccentricity of $v$
in $T$. A vertex $v$ is called an internal vertex if $C(v)\neq \emptyset$. Since the one-vertex tree is trivial, we assume that $|V(T)|\ge 2$. In particular,
every leaf of $T$ has positive eccentricity. Fix an integer $k\ge 1$, and let $[k]_0=\{0,1,2,\dots,k\}$.

\subsection{Boundary states}
\headinggap
 For a rooted subtree $T_v$, a vertex of $T_v$ is said to be dominated \emph{from inside $T_v$} if it is dominated by a broadcaster lying in $T_v$. We define two state values for every $v\in V(T)$ and every $a\in [k]_0$: the inside state $I_v(a)$ and the outside state $O_v(a)$.

We use the following convention. If a broadcaster $x$ of strength $f(x)$ dominates a vertex $y$, then
\[
f(x)-d_T(x,y)
\]
is the amount of broadcast strength still available when the broadcast reaches $y$.

\begin{definition}
For $a\in [k]_0$, let $I_v(a)$ denote the minimum cost of a $k$-limited broadcast assignment on $T_v$ such that every vertex of $T_v$ is dominated exactly once from inside $T_v$. In addition, if $x\in V(T_v)$ is the unique broadcaster that dominates the root $v$, then
\[
f(x)-d_T(x,v)=a.
\]
If no such assignment exists, we set $I_v(a)=\infty$.
\end{definition}

\begin{definition}
For $a\in [k]_0$, let $O_v(a)$ denote the minimum cost of a $k$-limited broadcast assignment on $T_v$ such that no vertex of $T_v$ at distance at most $a$ from $v$ is dominated from inside $T_v$, while every vertex of $T_v$ at distance greater than $a$ from $v$ is dominated exactly once from inside $T_v$. If no such assignment exists, we set $O_v(a)=\infty$.
\end{definition}

Thus $I_v(a)$ represents the case where the whole subtree $T_v$, including its root $v$, is already dominated exactly once by broadcasters inside $T_v$. The value $a$ records how much strength remains when the broadcast reaches the root $v$. On the other hand, $O_v(a)$ represents the case where the vertices of $T_v$ within distance at most $a$ from $v$ are left to be dominated from outside $T_v$, while the remaining vertices of $T_v$ are already dominated exactly once from inside.

At the root $r$ of the tree $T$, nothing can be left to vertices outside the tree. Hence only
the inside states can produce a $k$-ELDB of the whole tree, and therefore
\[
\gamma_{ebk}(T)=\min\{I_r(a):a\in [k]_0\}.
\]
Here an entry equal to $\infty$ means, as in the definitions above, that no
broadcast assignment satisfying the stated conditions exists. Thus $T$ admits a
$k$-ELDB if and only if the above minimum is finite.

\subsection{Initialization at the leaves}
\headinggap
The leaves provide the base cases.

\begin{lemma}\label{lem:leaf-init-tree}
If $v$ is a leaf of a nontrivial rooted tree $T$, then $O_v(a)=0$ for all $a\in [k]_0$, while
$I_v(0)=\infty$ and
\[
I_v(a)=
\begin{cases}
a, & \text{if } 1\le a\le \min\{k,\ecc_T(v)\},\\[4pt]
\infty, & \text{otherwise.}
\end{cases}
\]
\end{lemma}

\begin{proof}

If $v$ is a leaf, then $T_v$ consists only of the single vertex $v$. Hence $O_v(a)=0$, for all $a\in  [k]_0$. For the inside states,
the only possible broadcaster is $v$ itself, which yields cost $a$ whenever a
broadcast of strength $a$ is allowed.
\end{proof}

\subsection{Local recurrences}
\headinggap
We now compute the values $I_v(a)$ and $O_v(a)$ for an internal vertex $v$
from the already computed values of its children. The following observation
explains why the child subtrees can be combined independently.

\begin{observation}\label{obs:through-child-tree-short}
Let $u$ be a child of $v$, and let $x\in V(T_u)$. If a broadcaster $x$
dominates a vertex outside $T_u$, then it also dominates $u$.
\end{observation}

\begin{proof}
The result follows trivially since every path from $x$ to a vertex outside $T_u$ passes through $u$.
\end{proof}
Using this observation, we obtain the recurrence relations below. We begin with the outside states.

{
\begin{lemma}[Outside recurrence]\label{lem:outside-recurrence-tree}
Let $v$ be a vertex with $C(v)\neq \emptyset$. Then
$
O_v(0)=\sum_{u\in C(v)} I_u(0),
$
and for $1\le a\le k$,
$
O_v(a)=\sum_{u\in C(v)} O_u(a-1).
$
\end{lemma}

\begin{proof}
For $O_v(0)$, the vertex $v$ is left to be dominated from outside $T_v$,
while every vertex in each child subtree $T_u$ must already be dominated
exactly once from inside $T_u$. Moreover, the broadcast that dominates $u$
from inside $T_u$ cannot also reach $v$; otherwise $v$ would not be left to be dominated from
the outside. Hence the required state in $T_u$ is $I_u(0)$. Summing over all
child subtrees gives
$
O_v(0)=\sum_{u\in C(v)} I_u(0).
$

For $1\le a\le k$, the vertices of $T_v$ at distance at most $a$ from $v$
are left to be dominated from outside $T_v$. In a child subtree $T_u$, this
corresponds exactly to leaving the vertices at distance at most $a-1$ from
$u$ to the outside, while all remaining vertices of $T_u$ are dominated
exactly once from inside. Thus the child subtree contributes $O_u(a-1)$, and
the formula follows by summing over all children.
\end{proof}
}

{
For the inside states $I_v(a)$, the unique broadcaster that dominates $v$ is
either $v$ itself or lies in exactly one child subtree.
}

{
\begin{lemma}[Self-source contribution]\label{lem:self-source-tree}
In the computation of $I_v(a)$, suppose the unique broadcaster that dominates
$v$ is $v$ itself. Then the corresponding minimum possible cost is
\[
\operatorname{Self}_v(a)=
\begin{cases}
a+\displaystyle\sum_{u\in C(v)} O_u(a-1),
& \text{if }1\le a\le \min\{k,\ecc_T(v)\},\\[10pt]
\infty, & \text{otherwise.}
\end{cases}
\]
\end{lemma}

\begin{proof}
If $v$ dominates itself, then $v$ must broadcast with strength $a$. This is
possible only when $1\le a\le \min\{k,\ecc_T(v)\}$. In that case, for each
child $u$, the broadcast from $v$ dominates exactly the vertices of $T_u$
within distance at most $a-1$ from $u$. The remaining vertices of $T_u$
must therefore be dominated exactly once from inside $T_u$, which is the state
$O_u(a-1)$. Adding the cost $a$ at $v$ gives the stated formula.
\end{proof}
}

{
\begin{lemma}[Child-source contribution]\label{lem:child-source-tree}
In the computation of $I_v(a)$, suppose the unique broadcaster that dominates
$v$ lies in exactly one child subtree. Then the corresponding minimum possible cost is
\[
\operatorname{Child}_v(0)=
\min_{t\in C(v)}
\left(
I_t(1)+\sum_{u\in C(v)\setminus\{t\}} I_u(0)
\right),
\]
and for $1\le a<k$,
\[
\operatorname{Child}_v(a)=
\min_{t\in C(v)}
\left(
I_t(a+1)+\sum_{u\in C(v)\setminus\{t\}} O_u(a-1)
\right).
\]
Moreover, $\operatorname{Child}_v(k)=\infty$.
\end{lemma}

\begin{proof}
Assume that the broadcaster dominating $v$ lies in the child subtree $T_t$.
Since $t$ is adjacent to $v$, the same broadcast reaches $t$ with one more
unit of strength than it has at $v$. Hence $T_t$ contributes $I_t(a+1)$.

If $a=0$, the broadcast from $T_t$ reaches $v$ but does not enter any other
child subtree. Therefore each $T_u$, $u\neq t$, must dominate its own root
from inside without reaching $v$, which gives the term $I_u(0)$.

If $1\le a<k$, the broadcast also enters every other child subtree $T_u$ with
$a-1$ units of strength at $u$. Thus $T_u$ must leave the vertices within
distance $a-1$ from $u$ to this broadcast, and dominate all remaining
vertices from inside. This gives the term $O_u(a-1)$. Taking the minimum over
the choice of $t$ gives the formula. The case $a=k$ would require the state
$I_t(k+1)$, which is outside the allowed range.
\end{proof}
}

{
\begin{corollary}[Inside recurrence]\label{cor:inside-recurrence-tree}
For every vertex $v$ with $C(v)\neq\emptyset$ and every $a\in [k]_0$,
\[
I_v(a)=\min\{\operatorname{Self}_v(a),\operatorname{Child}_v(a)\}.
\]
\end{corollary}

\begin{proof}
In the state $I_v(a)$, the vertex $v$ is dominated by exactly one
broadcaster. This broadcaster is either $v$ itself or lies in exactly one child
subtree. Lemmas~\ref{lem:self-source-tree} and \ref{lem:child-source-tree}
compute the minimum cost in these two cases, hence the result follows.
\end{proof}
}

\subsection{The dynamic programming procedure}
\headinggap
The recurrence formulas above lead to a bottom-up algorithm. Root the
tree $T$ at a vertex $r$, and process the vertices in postorder. We store the computed values $I_v(a)$ and $O_v(a)$, for
$v\in V(T)$ and $a\in [k]_0$, in a dynamic programming table.

\medskip
\noindent
\emph{Algorithm: $\mathrm{DP\text{-}Tree}(T,k)$.}

\noindent
\emph{Input:} A tree $T$ and a fixed integer $k\ge 1$.

\noindent
\emph{Output:} The value $\gamma_{ebk}(T)$, or $\infty$ if $T$ admits no
$k$-ELDB.

\begin{enumerate}
    \item Root $T$ at an arbitrary vertex $r$, and list the vertices of $T$ in
    postorder.

    \item For each leaf $v$, initialize $I_v(a)$ and $O_v(a)$ for all
    $a\in [k]_0$ using Lemma~\ref{lem:leaf-init-tree}.

    \item For each non-leaf vertex $v$, compute $O_v(a)$ for all $a\in [k]_0$ using
    Lemma~\ref{lem:outside-recurrence-tree}.

    \item Compute $\operatorname{Self}_v(a)$ and
    $\operatorname{Child}_v(a)$ for all $a\in [k]_0$ using
    Lemmas~\ref{lem:self-source-tree} and~\ref{lem:child-source-tree}, and set
    \[
    I_v(a)=\min\{\operatorname{Self}_v(a),\operatorname{Child}_v(a)\}.
    \]

\item Return $\min\{I_r(a):a\in [k]_0\}$. If this value is finite, then it equals $\gamma_{ebk}(T)$; if it is $\infty$, then $T$ admits no $k$-ELDB.
    
\end{enumerate}

{
\begin{theorem}\label{thm:dp-correct-complexity}
Let $T$ be a tree on $n$ vertices, and let $k\ge 1$ be fixed. The algorithm
$\mathrm{DP\text{-}Tree}(T,k)$ computes $\gamma_{ebk}(T)$. In particular,
it decides whether $T$ admits a $k$-ELDB. Moreover, the algorithm runs in
$O(nk)$ time.
\end{theorem}

\begin{proof}
The correctness follows from the postorder computation. The values at the leaves
are initialized by Lemma~\ref{lem:leaf-init-tree}. Now suppose that
$v$ is a vertex whose children have already been processed, then the outside states
$O_v(a)$ are given by Lemma~\ref{lem:outside-recurrence-tree}. For the inside
states, the broadcaster that dominates $v$ is either $v$ itself or lies in
exactly one child subtree. These two cases are precisely the contributions
$\operatorname{Self}_v(a)$ and $\operatorname{Child}_v(a)$, and hence
Corollary~\ref{cor:inside-recurrence-tree} gives the correct value of $I_v(a)$.
 Thus, by induction along the postorder
traversal, the values $I_v(a)$ and $O_v(a)$ stored in the dynamic programming
table are computed correctly for every vertex $v$ and every $a\in [k]_0$.

At the root $r$, nothing can be left to vertices outside the tree. Therefore
\[
\gamma_{ebk}(T)=\min\{I_r(a):a\in [k]_0\}.
\]
If this minimum is finite, then $T$ admits a $k$-ELDB; if it is infinite, no
$k$-ELDB exists.

It remains to estimate the running time. For each vertex $v$, precompute
\[
S_v^I=\sum_{u\in C(v)}I_u(0)
\quad\text{and}\quad
S_v^O(b)=\sum_{u\in C(v)}O_u(b),\qquad 0\le b\le k-1.
\]
Using these sums, all values $O_v(a)$, $\operatorname{Self}_v(a)$, and
$\operatorname{Child}_v(a)$ can be computed by scanning the child list of $v$
for each $a\in [k]_0$. Hence the work at $v$ is $O(k|C(v)|)$. Since
\[
\sum_{v\in V(T)} |C(v)|=n-1,
\]
the total running time is $O(nk)$.
\end{proof}
}

\begin{corollary}
The parameter $\mcr(T)$ can be computed in polynomial time for every tree $T$.
\end{corollary}

\begin{proof}
Run $\mathrm{DP\text{-}Tree}(T,q)$ for each $q$, where
$q \in \{ 1,2,\dots,\operatorname{rad}(T)$\}, and choose the least value of $q$ for
which the algorithm returns a finite value. By Theorem~\ref{thm:dp-correct-complexity},
this happens exactly when $T$ admits a $q$-ELDB. Hence the first such value is
$\mcr(T)$.

Since the run for a fixed $q$ takes $O(nq)$ time, the total running time is
\[
\sum_{q=1}^{\operatorname{rad}(T)}O(nq)
=
O(n\operatorname{rad}(T)^2),
\]
which is polynomial in $n$. 
\end{proof}

\subsection{Reconstruction of an optimal broadcast}
\headinggap
{
The dynamic program may be augmented to return an optimal broadcast. Whenever a
minimum is computed in one of the recurrence formulas, we retain one choice at
which the minimum is attained. Thus, for an outside state, the corresponding
child states are determined by Lemma~\ref{lem:outside-recurrence-tree}. For an
inside state $I_v(a)$, we retain whether the minimum is obtained from
$\operatorname{Self}_v(a)$ or from $\operatorname{Child}_v(a)$. In the former
case, the vertex $v$ is assigned broadcast value $a$; in the latter case, we
retain one child $t\in C(v)$ attaining the minimum in
Lemma~\ref{lem:child-source-tree}.
}

{
After the dynamic program has been completed, choose $a^\ast\in [k]_0$ such
that
\[
I_r(a^\ast)=\min\{I_r(a):a\in [k]_0\}.
\]
Tracing the retained choices from the state $I_r(a^\ast)$ down the rooted tree
produces the broadcasting vertices and their assigned values.
}

\begin{theorem}\label{thm:reconstruction}
{
If $T$ admits a $k$-ELDB, then an optimal efficient $k$-limited dominating
broadcast of $T$ can be reconstructed in linear time after the values
$I_v(a)$ and $O_v(a)$ have been computed.
}
\end{theorem}

\begin{proof}
{
The retained choices specify, for each visited state, the corresponding states in
the child subtrees and, when applicable, the broadcast value assigned to the
current vertex. Starting from a minimum root state $I_r(a^\ast)$, this recursive
tracing follows the rooted-tree structure and visits each vertex at most once.
The resulting broadcast realizes the same choices used in the computation of
$\gamma_{ebk}(T)$, and hence has minimum cost. Therefore it is an optimal
$k$-ELDB. The reconstruction time is linear in $|V(T)|$.
}
\end{proof}

\section{Computational complexity}
\headinggap
We now turn to the complexity of the decision problem. For a fixed integer
$k\ge 1$, we consider the following decision problem.

\medskip
\noindent
\emph{The Efficient $k$-limited broadcast domination ($k$-ELDB) problem}

\noindent
\emph{Instance:} A graph $G$.

\noindent
\emph{Question:} Does $G$ admit an efficient $k$-limited dominating broadcast?

\medskip

The case $k=1$ is already familiar: an efficient $1$-limited dominating
broadcast is nothing but an efficient dominating set. Thus, the problem for
$k=1$ coincides with efficient domination, which is known to be NP-complete on
arbitrary graphs \cite{bange1988,haynes1998}. Our aim here is to show that the same
hardness persists for every fixed $k$, where $k\ge 2$.

\subsection{The truth gadget $T_k$}
\headinggap
\medskip
\noindent
We begin by describing the variable gadget that will be used in the reduction.
A vertex that is adjacent to a leaf vertex (vertex of degree $1$) is said to be a \emph{support vertex}.

\begin{observation}\label{obs_oTk_final}
If $f$ is an efficient dominating broadcast on a graph $G$, then no support
vertex of $G$ can lie at distance exactly $f(v)$ from a broadcasting vertex
$v$.
\end{observation}

Next if, for each $k\ge 1$, we constructively define a bicentral tree $T_k$ as follows.

\begin{construction}\label{cons:Tk_final} (Refer Figure \ref{fig:Tk_final})
For $k=1$, let $T_1\cong K_2$, with vertices $u$ and $\overline{u}$.
For $k\ge 2$, start with the edge $u\overline{u}$. Set
$x_0=u$ and $y_0=\overline{u}$.  Attach a path of length
$k-1$ at $u$ by adding vertices $x_1,\dots,x_{k-1}$ and edges
$ux_1,x_1x_2,\dots,x_{k-2}x_{k-1}$. Similarly, attach a path of length
$k-1$ at $\overline{u}$ by adding vertices $y_1,\dots,y_{k-1}$ and edges
$\overline{u}y_1,y_1y_2,\dots,y_{k-2}y_{k-1}$. Finally, for each $j$, $j \in \{0,1,\dots,k-2\}$,
attach one pendant leaf to $x_j$ and one pendant leaf to $y_j$.
\end{construction}

Thus, $u$ and $\overline{u}$ are the two central vertices of the gadget, while
$x_{k-1}$ and $y_{k-1}$ are the terminal leaves at the ends of the two main
branches.

\begin{figure}
    \centering
    \includegraphics[scale=0.27]{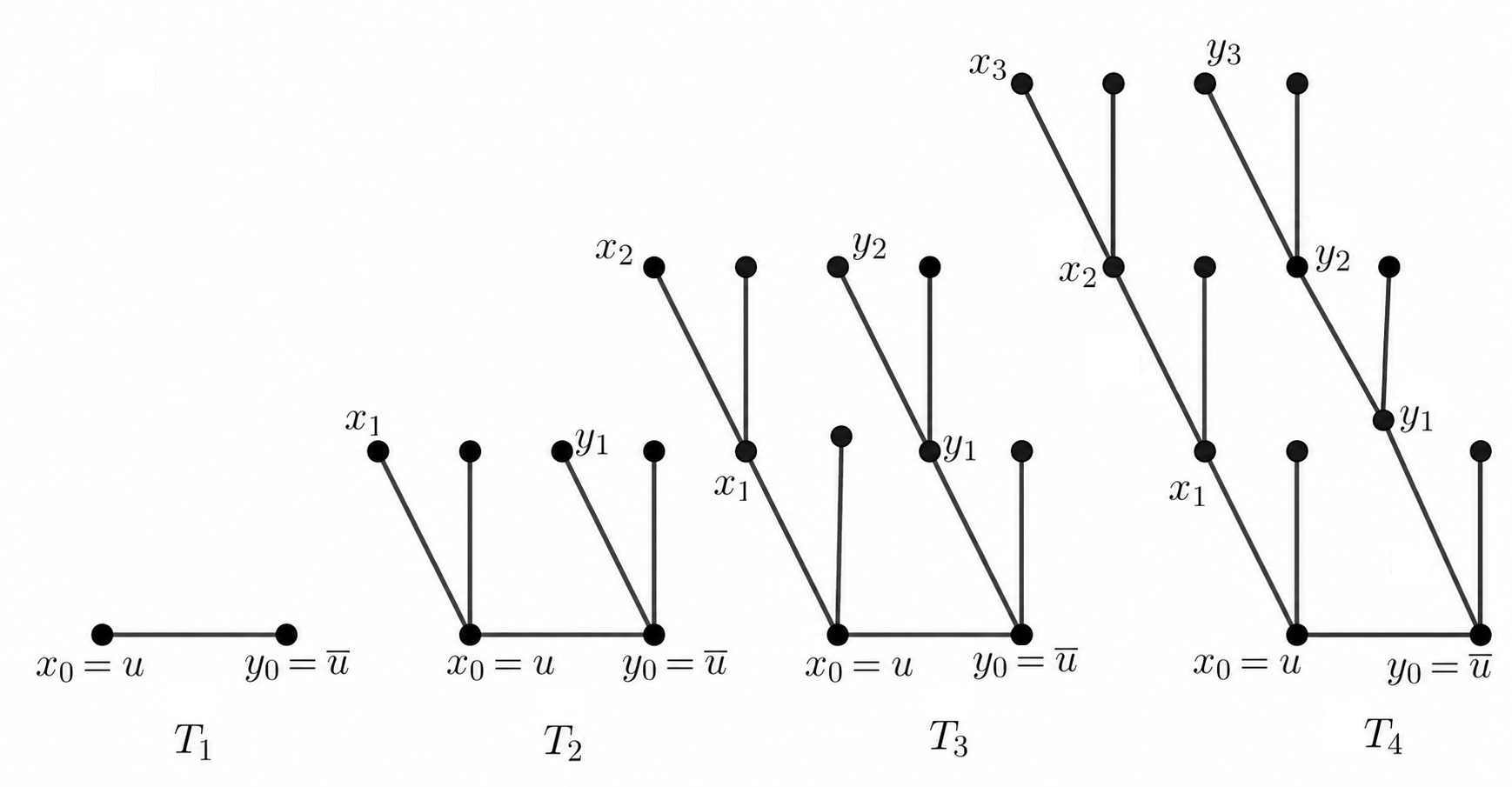}
    \caption{The truth gadget $T_k$.}
    \label{fig:Tk_final}
\end{figure}

The next lemma records the $k$-ELDBs of $T_k$ which will be used in the
reduction.

\begin{lemma}\label{lem:Tk-unique-final}
The graph $T_k$ admits exactly two efficient $k$-limited dominating
broadcasts. In each of them, exactly one of the two central vertices $u$ and
$\overline{u}$ receives the value $k$, and every other vertex receives the
value $0$.
\end{lemma}

\begin{proof}
For $k=1$, the graph $T_1\cong K_2$, and the result is immediate. Assume
$k\ge 2$. Assigning value $k$ to $u$, and $0$ to all other vertices, dominates all of $T_k$ exactly once; the same is true symmetrically for $\overline{u}$. Thus
the two stated broadcasts are $k$-ELDBs. Next, it remains to see that no other broadcast is possible. By the construction of $T_k$, every noncentral vertex has a support vertex at each possible positive
broadcast distance from it. Hence, by Observation~\ref{obs_oTk_final}, no
noncentral vertex can be a broadcasting vertex.

Therefore exactly one of $u$ and $\overline{u}$ is chosen as a broadcasting vertex and each forms a $k$-ELDB of $T-k$. Further, the strength of the chosen broadcasting vertex must be $k$ so as to dominate leaf vertex at a distance $k$ in the opposite branch. Hence, $T_k$ has exactly two efficient $k$-limited dominating broadcasts. 

\end{proof}

\begin{corollary}\label{cor:all-mcr-values}
For every integer $k\ge 1$, there exists a graph $G$ with $\mcr(G)=k$.
\end{corollary}

\begin{proof}
Take $G=T_k$. By Lemma~\ref{lem:Tk-unique-final}, $T_k$  has exactly two effficnet $k$-limited dominating broadcasts, and assign strength $k$ at one of the two central vertices. Suppose $T_k$ admits an efficient $\ell$-limited broadcast, say $f$, for some
$\ell<k$, then $f$ also a $k$-limited broadcast,
contradicting Lemma~\ref{lem:Tk-unique-final}. Hence $\mcr(T_k)=k$.
\end{proof}

\subsection{NP-completeness of the $k$-ELDB problem for fixed $k$}
\headinggap
We now prove that the $k$-ELDB problem is NP-complete for any given $k$.  The reduction is from the following classical NP-complete problem

\medskip
\noindent
\textsc{Exact $1$-in-$3$ SAT} 

\noindent
{
\emph{Instance:} A Boolean formula
\[
\Phi=C_1\wedge C_2\wedge\cdots\wedge C_m,
\]
where each clause $C_j$ contains exactly three literals. Each literal is either a variable $u_i$ or its negation $\overline{u_i}$.
}

\noindent
{
\emph{Question:} Does there exist a truth assignment
\[
\tau:\{u_1,u_2,\dots,u_n\}\to\{0,1\}
\]
such that, in every clause $C_j$, exactly one of the three literals is true?
}

It has been shown that the Exact $1$-in-$3$ SAT problem is NP-complete (\cite{schaefer1978}).

Fix $k\ge 2$. Let $\Phi=C_1\wedge C_2\wedge \cdots \wedge C_m$ be an instance
of Exact $1$-in-$3$ SAT with variables $u_1,u_2,\dots,u_n$. We construct a
graph $G_k(\Phi)$ as follows:

For each variable $u_i$, take a copy $T_k^i$ of the truth gadget $T_k$, and
denote its two central vertices by $u_i$ and $\overline{u_i}$. For each clause
$C_j$, add a clause vertex in $G_k(\Phi)$, again denoted by $C_j$. If a literal
$\ell\in\{u_i,\overline{u_i}\}$ appears in the clause $C_j$, then join $C_j$
to $\ell$ by a path of length $k$. We denote this path by $Q_{j,\ell}$, and
write $P_{j,\ell}=Q_{j,\ell}\setminus\{\ell\}$. Since $k$ is fixed, the graph
$G_k(\Phi)$ can be constructed in polynomial time. Note that, all clause--literal paths are taken to be internally vertex-disjoint.

The reverse implication of the reduction is driven by the following lemmas.

\begin{center}
\includegraphics[width=0.92\linewidth]{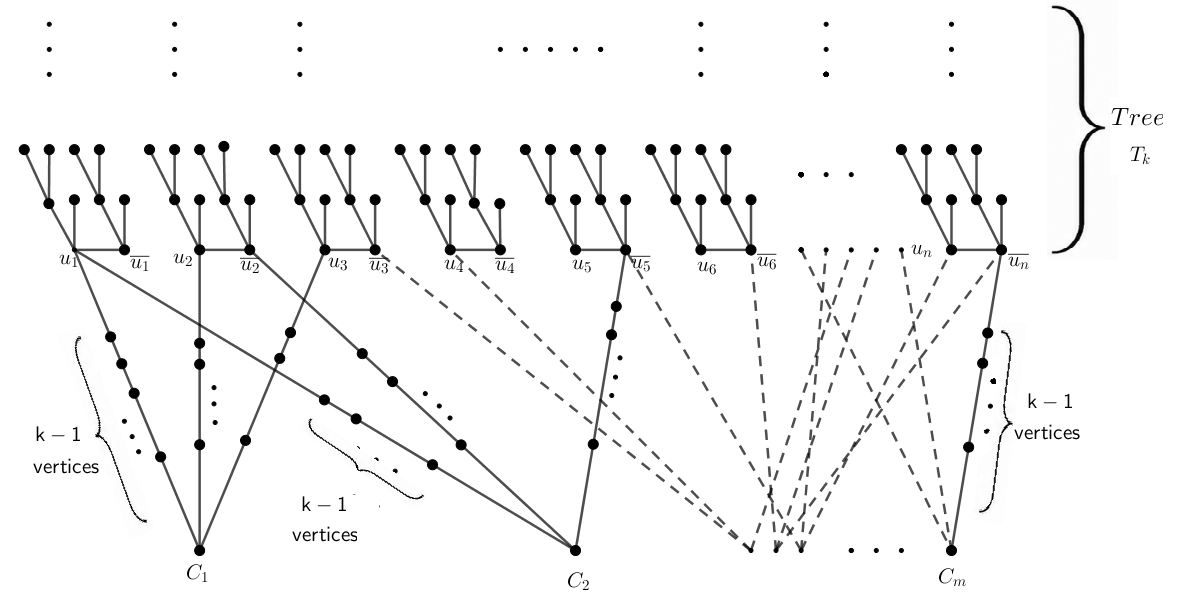}
\captionof{figure}{The reduction graph $G_k(\Phi)$ for a sample instance of Exact
$1$-in-$3$ SAT. When $k=1$, the clause--literal paths collapse to edges,
and the picture reduces to the familiar efficient domination construction.}
\label{fig:keds_final}
\end{center}

\begin{lemma}\label{lem:augmented-gadget-final}
Let $T_k^i$ be one of the variable gadgets in $G_k(\Phi)$. In any efficient
$k$-limited dominating broadcast of $G_k(\Phi)$,
\begin{enumerate}
    \item exactly one of the two vertices $u_i$ and $\overline{u_i}$
    broadcasts,
    \item the broadcast value at that vertex is $k$, and
    \item no other vertex of $T_k^i$ broadcasts.
\end{enumerate}
\end{lemma}

\begin{proof}
Fix a variable gadget $T_k^i$. We first show that the degree-one end vertices
$x_{k-1}$ and $y_{k-1}$ of $T_k^i$ must be dominated by broadcasting vertices
belonging to $T_k^i$.

The only vertices of $G_k(\Phi)$ not belonging to $T_k^i$ that lie within
distance $k$ from $x_{k-1}$ are the first internal vertices on
clause--literal paths attached at $u_i$. Such a vertex is at distance exactly
$k$ from the support vertex $y_{k-2}$ of $T_k^i$, and hence cannot be a
broadcasting vertex by Observation~\ref{obs_oTk_final}. Similarly, the first
internal vertices on clause--literal paths attached at $\overline{u_i}$ are at
distance exactly $k$ from the support vertex $x_{k-2}$, and so they cannot
dominate $y_{k-1}$. Therefore, $x_{k-1}$ and $y_{k-1}$ must be dominated by
broadcasting vertices belonging to $T_k^i$.

By the construction of $T_k^i$, every noncentral vertex of $T_k^i$ has a
support vertex at each possible positive broadcast distance from it. Hence, again
by Observation~\ref{obs_oTk_final}, no noncentral vertex of $T_k^i$ can be a
broadcasting vertex. Thus the only possible broadcasting vertices in $T_k^i$
are $u_i$ and $\overline{u_i}$.

At least one of $u_i$ and $\overline{u_i}$ must broadcast, since otherwise
the end vertices $x_{k-1}$ and $y_{k-1}$ would remain undominated. They cannot
both broadcast, because they are adjacent and would violate efficient property.
Finally, the selected central vertex must reach the end vertex at the far end of
the opposite branch, which is at distance $k$. Hence its broadcast value must
be $k$.
\end{proof}

{
For a literal $\ell\in\{u_i,\overline{u_i}\}$, the other central vertex in the
same variable gadget is called the \emph{complementary literal} of $\ell$.
Thus the complementary literal of $u_i$ is $\overline{u_i}$, and the
complementary literal of $\overline{u_i}$ is $u_i$. These two vertices are
adjacent in $T_k^i$.
}

\begin{lemma}\label{lem:no-path-broadcast-final2}
No internal vertex of any clause--literal path $Q_{j,\ell}$ can be a
broadcasting vertex in an efficient $k$-limited dominating broadcast of
$G_k(\Phi)$.
\end{lemma}

\begin{proof}

Let $z$ be an internal vertex of a clause--literal path $Q_{j,\ell}$. By
Lemma~\ref{lem:augmented-gadget-final}, exactly one of the two central vertices in
the corresponding variable gadget broadcasts with value $k$. If the selected
broadcasting vertex is $\ell$, then every vertex on the path $Q_{j,\ell}$,
including $z$, lies within distance at most $k$ from $\ell$, and hence is
already dominated.

If the selected broadcasting vertex is the complementary literal of $\ell$, then
$z$ is still within distance at most $k$ from that vertex. Indeed, the
complementary literal is adjacent to $\ell$ inside the gadget, and since $z$
is an internal vertex of the length-$k$ path from $\ell$ to $C_j$, we have
$d(\ell,z)\le k-1$. Therefore $z$ is already dominated in either case. Since
a broadcasting vertex always dominates itself, $z$ cannot broadcast without
violating efficiency.

\end{proof}

\begin{lemma}\label{lem:no-clause-broadcast-final2}
No clause vertex of $G_k(\Phi)$ can be a broadcasting vertex in an efficient
$k$-limited dominating broadcast of $G_k(\Phi)$.
\end{lemma}

\begin{proof}
Let $C_j$ be a clause vertex. Its three neighbors on the incident
clause--literal paths are already dominated by the broadcasting vertices by
Lemma~\ref{lem:augmented-gadget-final}. If $C_j$ were also to broadcast, then
those neighbors would be dominated twice, which is impossible in an efficient
broadcast. Hence, clause vertices do not broadcast.
\end{proof}

\begin{lemma}\label{lem:clause-exact-final2}
Let $
C_j=\ell_{j,1}\vee \ell_{j,2}\vee \ell_{j,3}
$ be a clause. In any efficient $k$-limited dominating broadcast of $G_k(\Phi)$,
the only vertices that can dominate $C_j$ are the three literal vertices
$\ell_{j,1},\ell_{j,2},\ell_{j,3}$. Consequently, $C_j$ is dominated exactly
once if and only if exactly one of those three literal vertices broadcasts.
\end{lemma}

\begin{proof}
By construction, each literal vertex $\ell_{j,t}$ lies at distance exactly $k$
from the clause vertex $C_j$. Thus if $\ell_{j,t}$ broadcasts with value $k$,
then it dominates $C_j$.

The complementary literal $\overline{\ell_{j,t}}$ is adjacent to $\ell_{j,t}$
inside the corresponding variable gadget. Hence,
$
d(C_j,\overline{\ell_{j,t}})=k+1,
$ and so the complementary literal cannot dominate $C_j$ in a $k$-limited
broadcast.

Now consider any other broadcasting vertex. By
Lemmas~\ref{lem:no-path-broadcast-final2} and
\ref{lem:no-clause-broadcast-final2}, neither an internal  vertex of path $Q_{j,\ell}$ nor the
clause vertex can broadcast. The only remaining candidates are central
vertices of variable gadgets. If such a vertex is not one of the three literal
vertices appearing in the clause, then every path from that vertex to $C_j$ has
length strictly greater than $k$, and so it cannot dominate $C_j$.

Therefore, the only  vertices that can dominate $C_j$ are precisely
the three literal vertices appearing in the clause. Since the broadcast is
efficient, $C_j$ must be dominated exactly once, and this happens exactly when
one of those three literals broadcasts and the other two do not.
\end{proof}

\begin{theorem}\label{thm:fixed-k-final}
For every fixed integer $k\ge 1$, the $k$-ELDB problem is NP-complete for arbitrary graphs.
\end{theorem}

\begin{proof}
Membership in NP is immediate.

For $k=1$, the problem is precisely the efficient domination problem, since a
$1$-ELDB is the same as an efficient dominating set, and is known to be  NP-complete. \cite{bange1988,haynes1998}.

It remains to consider the case $k\ge 2$. For this fixed value of $k$, we
reduce the $k$-ELDB problem from Exact $1$-in-$3$ SAT problem. Let $\Phi$ be an instance of Exact
$1$-in-$3$ SAT, and let $G_k(\Phi)$ be the graph constructed above.

Suppose first that $\Phi$ has a truth assignment under which each clause
contains exactly one true literal. Define a broadcast $f$ on $G_k(\Phi)$ as follows. For each literal vertex
$\ell\in\{u_i,\overline{u_i}\}$, set
\[
f(\ell)=
\begin{cases}
k, & \text{if the literal $\ell$ is true under the given truth assignment},\\
0, & \text{if the literal $\ell$ is false under the given truth assignment}.
\end{cases}
\]
For every other vertex $v$, set $f(v)=0$.

Inside each variable gadget, the chosen central vertex dominates the gadget
exactly once. Along every clause--literal path, the internal vertices are dominated
by the selected broadcaster in the corresponding variable gadget. Finally, because
each clause contains exactly one true literal, each clause vertex is dominated by
exactly one of its three incident literals. Thus, $f$ is an efficient
$k$-limited dominating broadcast of $G_k(\Phi)$.

Conversely, suppose that $G_k(\Phi)$ admits an efficient $k$-limited
dominating broadcast $f$. By Lemma~\ref{lem:augmented-gadget-final}, for each
variable $u_i$, exactly one of $u_i$ and $\overline{u_i}$ broadcasts with
value $k$. We therefore define a truth assignment $T$ by
\[
T(u_i)=1 \quad \Longleftrightarrow \quad f(u_i)=k,
\] or equivalently,
$
T(u_i)=0 \quad \Longleftrightarrow \quad f(\overline{u_i})=k.
$

Now, let
$
C_j=\ell_{j,1}\vee \ell_{j,2}\vee \ell_{j,3}
$
be any clause. By Lemma~\ref{lem:clause-exact-final2}, a clause vertex $C_j$
is dominated exactly once if and only if exactly one of the three literal
vertices $\ell_{j,1},\ell_{j,2},\ell_{j,3}$ broadcasts. By the definition of the
truth assignment $T$, this is equivalent to saying that exactly one of the three
literals in the clause is true under $T$. Since $C_j$ is arbitrary, every
clause of $\Phi$ has exactly one true literal. Thus, $T$ is a satisfying
assignment for the Exact $1$-in-$3$ SAT instance.Therefore, we have shown that
\[
\Phi \text{ is satisfiable } \iff G_k(\Phi) \text{ admits an efficient }
k\text{-limited dominating broadcast}.
\]
Since $k$ is fixed, the construction of $G_k(\Phi)$ is polynomial in the size
of $\Phi$. Hence the efficient $k$-limited broadcast domination problem is
NP-complete for every fixed $k$, where $k\ge 2$. Together with the case $k=1$, this proves
the theorem.
\end{proof}

\begin{corollary}\label{cor:compute-mcr-final}
Computing $\mcr(G)$ is NP-hard for arbitrary graphs.
\end{corollary}

\begin{proof}
Suppose that $mcr(G)$ could be computed in polynomial time. Then, for any fixed
$k\ge 1$, we could decide in polynomial time whether a graph admits an efficient
$k$-limited dominating broadcast simply by computing $mcr(G)$ and consequently checking
whether $mcr(G)\le k$. This contradicts Theorem~\ref{thm:fixed-k-final}
unless $\mathrm{P}=\mathrm{NP}$. Therefore computing $mcr(G)$ is NP-hard.
\end{proof}

\begin{remark}
This reduction is compatible with the classical reduction for the efficient domination problem from Exact $1$-in-$3$ SAT. When  $k=1 $, the truth gadget is the edge  $u_i\overline{u_i} $, and the
clause--literal paths are ordinary edges. Thus, the construction specializes to
the usual variable-pair and clause-incidence structure used in reductions for
efficient domination. 
\end{remark}
\section{Conclusion and Future work}
\headinggap
We have studied efficient $k$-limited dominating broadcasts from two complementary viewpoints. On one hand, we proved that the $k$-ELDB problem is NP-complete for arbitrary graphs. On the other hand, we showed that trees admit a polynomial-time dynamic programming algorithm. By encoding the behaviour of rooted subtrees through a finite collection of boundary states, one can compute $\gamma_{ebk}(T)$ for fixed $k$, and hence determine $\mcr(T)$.

The tree case illustrates particularly well how local structure can be exploited in the presence of a global efficiency condition. Although efficient domination itself is already quite rigid, allowing bounded broadcast strength creates a richer hierarchy of possibilities, and the parameter $\mcr(T)$ captures this transition in a natural way. 

Among the various directions in which the present work could be continued, one natural question is whether the state-based method can be extended beyond trees, for instance to block graphs, line graphs of trees, or other graph classes with a useful decomposition structure. On the complexity side, it would also be interesting to understand the status of the problem on intermediate classes such as bipartite, chordal, or planar graphs, thus paving the way for a broader study of the efficient $k$-limited broadcast domination problem.


\begin{thebibliography}{00}

\bibitem{bange1988}
D.W. Bange, A.E. Barkauskas, P.J. Slater,
Efficient dominating sets in graphs,
in: R.D. Ringeisen, F.S. Roberts (Eds.),
Applications of Discrete Mathematics,
SIAM, Philadelphia, 1988, pp. 189--199.

\bibitem{caceres2018}
J. C{\'a}ceres, C. Hernando, M. Mora, I.M. Pelayo, M.L. Puertas,
General bounds on limited broadcast domination,
Discrete Math. Theor. Comput. Sci. 20 (2) (2018), Article 13.

\bibitem{cockayne2011broadcasts}
E.J. Cockayne, S. Herke, C.M. Mynhardt,
Broadcasts and domination in trees,
Discrete Math. 311 (13) (2011) 1235--1246.

\bibitem{dunbar2006broadcasts}
J.E. Dunbar, D.J. Erwin, T.W. Haynes, S.M. Hedetniemi, S.T. Hedetniemi,
Broadcasts in graphs,
Discrete Appl. Math. 154 (1) (2006) 59--75.

\bibitem{erwin2001cost}
D.J. Erwin,
Cost domination in graphs,
Ph.D. thesis, Western Michigan University, Kalamazoo, MI, 2001.

\bibitem{haynes1998}
T.W. Haynes, S.T. Hedetniemi, P.J. Slater,
Fundamentals of Domination in Graphs,
Marcel Dekker, New York, 1998.

\bibitem{haynes2021structures}
T.W. Haynes, S.T. Hedetniemi, M.A. Henning (Eds.),
Structures of Domination in Graphs,
Developments in Mathematics, Vol. 66,
Springer, Cham, 2021.

\bibitem{heggernes2006}
P. Heggernes, D. Lokshtanov,
Optimal broadcast domination in polynomial time,
Discrete Math. 306 (24) (2006) 3267--3280.

\bibitem{herke2009dominating}
S.R.A. Herke,
Dominating broadcasts in graphs,
Ph.D. thesis, University of Victoria, Victoria, BC, 2009.

\bibitem{herke2009radial}
S.R.A. Herke, C.M. Mynhardt,
Radial trees,
Discrete Math. 309 (20) (2009) 5950--5962.

\bibitem{lunney2015more}
S. Lunney, C.M. Mynhardt,
More trees with equal broadcast and domination numbers,
Australas. J. Combin. 61 (2015) 251--272.

\bibitem{mitchell1979linear}
S.L. Mitchell, E.J. Cockayne, S.T. Hedetniemi,
Linear algorithms on recursive representations of trees,
J. Comput. Syst. Sci. 18 (1) (1979) 76--85.

\bibitem{mynhardt2013class}
C.M. Mynhardt, J. Wodlinger,
A class of trees with equal broadcast and domination numbers,
Australas. J. Combin. 56 (2013) 3--22.

\bibitem{schaefer1978}
T.J. Schaefer,
The complexity of satisfiability problems,
in: Proceedings of the Tenth Annual ACM Symposium on Theory of Computing,
ACM, New York, 1978, pp. 216--226.

\bibitem{seager2008dominating}
S.M. Seager,
Dominating broadcast of caterpillars,
Ars Combin. 88 (2008) 307--319.

\end{thebibliography}
\end{document}